\documentclass[runningheads]{llncs}
\usepackage{graphicx}
\usepackage{amsmath,amssymb,amsfonts}
\usepackage{graphicx}
\usepackage{textcomp}
\usepackage{xcolor}
\usepackage{balance}
\usepackage{fixmath}
\usepackage{framed}
\usepackage{mathtools}
\usepackage{pgf,tikz}
\usepackage{array} 
\usepackage{diagbox}
\usepackage{pict2e} 
\usepackage{graphicx}
\usepackage{algorithm}
\usepackage{algorithmicx}
\usepackage{algpseudocode}

\usepackage{listings}
\usepackage{xspace}
\usepackage{hyperref}
\usepackage{amsfonts}
\usepackage{amsmath}
\usepackage[colorinlistoftodos]{todonotes} %sergio for todo?
\newcommand{\var}[1]{\text{\lstinline+#1+}}

\def\cA{\ensuremath{{\mathcal A}}}      \def\cC{\ensuremath{{\mathcal C}}}   
         
\def\cI{\ensuremath{{\mathcal I}}}         
      \def\cO{\ensuremath{{\mathcal O}}}   \def\cP{\ensuremath{{\mathcal P}}}
\def\cQ{\ensuremath{{\mathcal Q}}}      
\def\cT{\ensuremath{{\mathcal T}}}

\newcommand{\set}[1]{\left\{ #1 \right\}}

\usepackage{epsfig}

\newcommand{\eqnlabel}[1]{\label{eq:#1}}
\newcommand{\nakedeqnref}[1]{\ref{eq:#1}}
\newcommand{\eqnref}[1]{Equation~\nakedeqnref{#1}}

\newcommand{\lemmalabel}[1]{\label{lemma:#1}}
\newcommand{\nakedlemmaref}[1]{\ref{lemma:#1}}
\newcommand{\lemmaref}[1]{Lemma~\nakedlemmaref{#1}}

\newcommand{\seclabel}[1]{\label{sec:#1}}
\newcommand{\nakedsecref}[1]{\ref{sec:#1}}
\newcommand{\secref}[1]{Section~\nakedsecref{#1}}

\newcommand{\Reals}{\ensuremath{\mathbb{R}}}
\graphicspath{{./pictures/}}

\title{The Smart Contract Model}

\author{
Yackolley Amoussou-Guenou\inst{1}
\and 
Maurice Herlihy \inst{2} \orcidID{0000-0002-3059-8926} \thanks{Part of this work was done while visiting LIP6, France, supported by the CNRS Fellow Ambassadeur program.}
\and
Sucharita Jayanti \inst{2}
\and
Maria Potop-Butucaru\inst{3} 
\and
Sergio Rajsbaum\inst{4} \orcidID{0000-0002-0009-5287} \thanks{Part of this work was done while visiting IRIF and LIP6, France.}
}

\institute{
Universit\'e Paris-Panth\'eon-Assas, CRED, Paris, France
\and
Brown University Computer Science Dept, Providence RI 02912, USA
\and
Sorbonne Universit\'e, LIP6, Paris, France
\and
Instituto de Mathem\'aticas, UNAM, Mexico.
}

\date{}

\begin{document}

\maketitle
\begin{abstract}
Many of the problems that arise in the context of blockchains and decentralized
finance can be seen as variations on classical problems of distributed computing.
The \emph{smart contract model} proposed here is intended to capture both the
similarities and the differences between classical and blockchain-based models of distributed computing.
The focus is on \emph{cross-chain} protocols
in which a collection of parties,
some honest and some perhaps not,
interact through trusted smart contracts residing on multiple, independent ledgers.

While cross-chain protocols are capable of general computations,
they are primarily used to track ownership of assets
such as cryptocurrencies or other valuable data.
For this reason, the smart contract model differs in some essential ways from
familiar models of distributed and concurrent computing.
Because parties are potentially Byzantine,
tasks to be solved are formulated using elementary game-theoretic notions,
taking into account the utility to each party of each possible outcome.
As in the classical model,
the parties provide task inputs and agree on a desired sequence of
proposed asset transfers.
Unlike the classical model,
the contracts, not the parties, determine task outputs in the form of executed asset transfers,
since they alone have the power to control ownership.

\end{abstract}

\section{Introduction}
\seclabel{intro}
The rise of decentralized finance and distributed ledger
technology presents both an opportunity and a challenge
to the distributed computing community.

The \emph{opportunity}: decentralized finance based on blockchains
has become a multi-billion dollar industry,
providing what seems to be an ideal application for the kinds of
fault-tolerant distributed algorithms and protocols
the distributed computing community has studied for decades.
Specific companies and cryptocurrencies may come and go,
but the ancient problem of how mutually untrusting parties can securely and efficiently exchange
valuable assets will only grow in importance as the world becomes more interconnected.

The \emph{challenge}:
classical 
 shared-memory and message-passing models
do not accurately mirror the realities of blockchains and smart contracts.
This paper proposes the \emph{smart contract model},
which we hope will help in adapting protocols and algorithms
from classical distributed computing to the world of decentralized, tamper-proof ledgers.
Our focus is on \emph{cross-chain} protocols
executed by a collection of active parties, some who follow the protocol, and some who might not.
interacting through shared smart contracts deployed on different blockchains. 

The smart contract model borrows from classical distributed computing models,
but it differs in some essential ways.
For example, in classical models,
a task is specified by its inputs and outputs.
Each process chooses an input,
and after communicating with the others via shared memory or messages,
it chooses an output consistent with the task specification.
In shared-memory models, 
shared objects typically keep track of a protocol's state as it executes.

By contrast,
while cross-chain protocols are capable of general computations,
they are primarily used to track ownership of assets
such as cryptocurrencies or other valuable data with a digital representation.
Smart contracts implement the ledgers that control
asset ownership.
Contracts are trusted, deterministic automata that change state in response
to authenticated messages from parties.

To carry out a task,
the parties agree on \emph{protocol}:
a sequence of steps for exchanging their assets.
\emph{Compliant} parties follow the agreed-upon protocol,
while \emph{deviating} parties can depart from the protocol in arbitrary ways.
No party can assume another will be compliant,
so any protocol must be designed to ensure both \emph{safety}:
no compliant party ends up worse off than it started,
and \emph{liveness}:
if all parties comply, all parties end up better off.

Because parties are autonomous (not under the protocol's control)
and potentially deviating,
tasks to be solved are formulated using elementary game-theoretic notions,
expressed in terms of the utility to each party of each possible outcome.
As in the classical model, the parties provide the inputs,
typically expressed as an agreed-upon sequence of proposed asset transfers.
Because some parties may deviate,
the transfers that actually take place may
differ from those agreed upon.
In the end, unlike in the classical model,
the \emph{contracts}, not the parties,
decide task outputs in the form of executed asset transfers,
since the contracts alone are trusted to control ownership of valuable assets.

The contracts are similar to shared objects in shared-memory distributed computing.
They are  deterministic state machines, 
but unlike the parties,
they are trustworthy.
Contract state and contract code are observable by any party,
since that state is replicated on multiple public repositories.
For example,
in an \emph{atomic cross-chain swap} task,
where two parties exchange assets that
reside on distinct ledgers.
The protocol must guarantee that if both parties conform to the protocol, 
then the swap takes place.
If one party deviates,
the other conforming party ends up no worse off than it started.

\section{Related Work}
\seclabel{related}
Various papers have described formal models of various types of ledgers.
Our focus differs in two ways. 
First, we are interested in the behavior of networks of smart contracts,
and view the ledger itself as a ``black box'' platform on top of which the contracts
are implemented.  
In this model, 
smart contracts exist on their own.
The contracts form a network of state machines,
where each contract exports a set of \emph{functions} (methods) called by parties,
operating on contract stae.
Each function's code is visible to all parties, and each party trusts the code will be executed correctly.
It is immaterial at this level of abstraction
if the smart contracts are implemented by a blockchain,
by a tamper-proof replicated database,
or some other future technology.
Second, our focus is on \emph{cross-chain transactions}, meaning that each smart contract runs on a different ledger.

Perhaps the first formal model of a permissionless blockchain system was proposed by Garay et al.~\cite{GarayKL2015,GarayKL2015-j2024} for Bitcoin. 
Later papers
aimed at modeling at a higher level of abstraction.
Anceaume et al.~\cite{AnceaumeLPT17} were the first to make a connection between the Bitcoin ledger and classical distributed shared objects. They introduced the notion of  a \emph{ Distributed Ledger Register} (DLR) where the value of the register has a tree topology instead of a single value as in the classical theory of distributed registers. 
The vertices of the tree are blocks of transactions cryptographically linked. 
%The DLR properties were crafted to fit the behavior of permissionless blockchains such as Bitcoin and Ethereum. 
Later, Anta et al.~\cite{AntaKGN18} introduced the {\em Distributed Ledger Object} (DLO) formalism that  defines the ledger object as an ordered sequence of records, abstracting away from registers.  This work has been extended~\cite{AntaGN19} to Multi-Distributed Ledger Objects (MDLO), the result of aggregating multiple Distributed Ledger Objects - DLO (a DLO is a formalization of the blockchain) and that supports append and get operations of records (e.g., transactions)  from multiple clients concurrently.
To model the behavior of distributed ledgers at runtime, Anceaume et al.\cite{AnceaumePLPP19} introduced the {\em Blockchain Abstract Data Type} abstraction, which provides a lower-level abstraction of Distributed Ledgers, suitable for both permissioned and permissionless systems.
Rajsbaum and Raynal~\cite{RR19cacm} provide a discussion of how ledgers are the next logical step in a long history of implementations of a sequential specification on a concurrent system. 
Other lines of research include Zappal{\`{a} et al.~\cite{ZappalaBPS21} who propose a game theoretical framework to formally characterize the robustness of blockchains systems in terms of resilience to rational deviations and immunity to Byzantine behaviors. The framework is sufficiently general to characterize the robustness of various blockchain protocols (e.g.  Bitcoin, Tendermint, Lightning Network, a side-chain protocols and a cross-chain swap protocols).  
Our model goes beyond  other models in the literature limited to simple payments~\cite{FreyGRT2021,GuerraouiKMPS2018,SliwinskiW2019} or swaps \cite{BelottiMPS20}.

Other research investigates formal models for blockchain-based contracts (overview in Bartoletti \emph{et al}~\cite{BBACZ21}).
In contrast to our work, they
address formal verification in the context of specific chains, such as Bitcoin or Ethereum ,
while our goal is a model reasonably independent of any particular system.

As a precursor to cross-chain coordination,
in the classic \emph{fair exchange} problem,
Alice and Bob wish to exchange digital assets
in such a way that either the exchange takes place,
or each party keeps its token.
Fair exchange has been widely studied.
It encompasses other important problems, such as contract signing,
and has been also studied for more than two participants.
It is well-known that this problem is unsolvable without a trusted third party~\cite{Pagnia1999OnTI}
because of its relation to consensus.
Fair exchange protocols remain relevant because smart contracts can be viewed
as trusted third parties.
See Tas et al.~\cite{ErtemSeresZMKBN2024} for references and applications to data storage.

The swap protocol of \secref{example} is synchronous.
Lys et al.~\cite{LysMP21} describe an alternative protocol
in a partially synchronous setting.

\section{Rationale}
\seclabel{real}
This section presents an informal description
how cross-chain coordination works in practice,
with the goal of motivating the formal model presented in the next section.

A \emph{distributed ledger}\footnote{
Following common usage, we use ``ledgers'', ``chains'', and ``blockchains''
interchangeably in informal discourse,
even though ledgers need not be implemented using blockchain technology.
} is a publicly-readable tamper-proof distributed database
used to track ownership of  \emph{assets},
which may be cryptocurrencies,
financial instruments, tokens, concert tickets,
or any data of value.
%Our discussion here is independent of which technology is used to implement a ledger.

\emph{Parties} own \emph{assets}.
A party can be a person, an organization, or even a bot acting on another party's behalf.
Parties are untrustworthy:
they may depart from any agreed-upon protocol in arbitrary ways,
including irrational ways that work against their own interests.
We assume only that every interaction includes at least one honest party who faithfully follows the agreed-upon protocol. 
While it is reasonable to assume, for example,
that a super-majority of the validators securing a blockchain are honest,
it is not reasonable (or prudent) to assume the same for a small number
of parties willing to participate in a one-shot financial transaction or other distributed task.
 
Parties may own assets controlled by one or more ledgers.
Parties are active agents, initiating and reacting to communication,
although parties do no typically communicate directly with other parties.
Instead,
interactions among parties are mediated by \emph{smart contracts}.
Each smart contract resides on a single ledger,
and different smart contracts may reside on different ledgers.
Contracts are automata that change state
in response to authenticated messages from parties. These messages are  structured as calls to functions exported by contracts.
The code of a smart contract must be deterministic, because it is repeatedly re-executed by mutually-suspicious parties
to check the correctness of earlier executions.\footnote{Even so-called ``probabilistic smart contracts''~\cite{ChatterjeeGP2019}
are actually deterministic because once they are written to the ledger,
they return the same results when re-executed.}

A smart contract's code and current state are public,
so a party calling a contract knows what code will be executed,
and trusts the contract to execute that code correctly.\footnote{Concurrent
calls to the same contract may be executed in
a non-deterministic order.}
Since the state of a smart contract is public,
a party can read a contract's state at any time.

In accordance with current practice,  contracts are passive:
a contract on one ledger cannot send messages directly to  a contract on another ledger
(because such network communications cannot be deterministically replayed). 
A contract $A$ on one chain can learn of a state change of
a contract $B$ on another chain only if some party explicitly
informs $A$ of $B$'s change,
perhaps accompanied by evidence that the informing party is honest.
For ease of exposition and without loss of generality,
we will assume that distinct smart contracts reside
on distinct ledgers.

A  \emph{cross-chain task} (or \emph{task} for short) specifies the problem that   multiple parties want to solve.
For example,
Alice, Bob, and Carol might agree to do a three-way token swap,
where each token is managed on a distinct ledger.
Alice will transfer token $A$ to Bob,
Bob will transfer token $B$ to Carol,
and Carol will transfer token $C$ to Alice.

A \emph{cross-chain protocol} is a sequence of actions for each party  to solve a cross-chain task.
Honest parties who follow the protocol are said to be \emph{compliant},
while dishonest parties who do not follow the protocol are said to be \emph{deviating}.
A party's compliance can often be monitored by the other parties,
but never enforced.
Unlike most conventional agreement protocols in Byzantine models,
we do not assume that a supermajority of participants is honest.
Instead,
we assume only that there is at least one compliant party in every protocol execution
(because
executions without compliant parties are not interesting).
While it is reasonable to assume, for example,
that a super-majority of the validators securing a blockchain are honest,
it is not reasonable (or prudent) to assume the same for a small number
of parties willing to participate in a particular financial transaction.

In our swap example,
Alice, Bob, and Carol take turns placing their assets into escrow
(by calling the distinct contracts managing the distinct asset types),
and when all assets are safely escrowed,
each party triggers a transfer from its counterparty. 
If some party deviates from the protocol,
then either all transfers are canceled,
or the deviating party loses its own asset.
See~\cite{Herlihy2018} for a complete protocol description,
and~\secref{example} for a two-party swap protocol.

The swap task's \emph{input party vector} describes who wants to transfer what to whom.
It makes sense for these inputs to be controlled by the participating parties. 
The task's \emph{input contract vector} describes who who owns which assets at the start.
The task's \emph{output contract vector} describes how asset ownership has changed:
who transferred what to whom.
It makes sense for these outputs to be controlled by the  contracts,
not the parties.
(In real life, banks, not their customers,
have the final word on customer account balances.)

What does it mean for a protocol to correctly implement a task?
By analogy with mainstream distributed computing,
we might naively require a correct protocol to be \emph{all-or-nothing}:
either all intended transfers take place, or none does,
a property also known as \emph{atomicity}.
While doing nothing should always be a legitimate course of action when things go wrong,
it may not always be possible.
Because parties can depart arbitrarily (and even irrationally) from the protocol,
we cannot prevent, for example, 
Alice from sending Bob two tokens instead of one.
Bob, if he is rational, will pocket the extra token and continue the protocol.
(Perhaps Alice and Bob are secretly laundering assets.)
In the end, Alice has one token less than the task specifies,
and Bob has one token more.
Here is a protocol outcome that is neither ``all'' nor ``nothing'',
yet this outcome is perfectly acceptable to the conforming party Bob.

It follows that when defining correctness in the cross-chain world,
it will be helpful to borrow
some elementary notions from game theory.

A party's \emph{utility} in an execution is a quantitative
measure of how much that party considers itself to be better
or worse of after the execution.
A configuration
is in a \emph{Nash equilibrium} if no party can increase its
utility by deviating while the other parties remain compliant.
Ensuring that a protocol is in Nash equilibrium is
necessary, but not sufficient.
A \emph{Sybil attack} occurs when one individual or organization
secretly controls a \emph{coalition} of multiple parties.
The controller may be willing to incur a small loss of utility
for one coalition member in return for a larger gain for another.
For this reason,
common sense demands that
protocols should be in Nash equilibrium across coalitions:
no coalition can increase its collective utility by deviating
while the other parties remain compliant\footnote{
This condition is not the same as a \emph{strong Nash equilibrium}~\cite{StrongNashEquilibrium},
which requires \emph{all} coalition members to gain utility by deviating.}.

No rational party would agree to participate in a protocol
that did not satisfy the following common-sense conditions:
\begin{itemize}
\item
  \emph{Coalition Nash Equilibrium}:
  No coalition of parties can increase its collective utility by deviating from the protocol
  while the others comply.

\item
  \emph{Liveness}:
  If all parties are compliant,
  then all and only agreed-upon asset transfers take place.
  
\item
  \emph{Safety}: No compliant party can end up ``worse off'',
  even if other parties deviate from the protocol\footnote{
  We treat as negligible costs in transaction fees and wasted time incurred by failed exchanges.
  In practice,
  some protocols compensate jilted parties~\cite{XueH2021} while others do not.}.
\end{itemize}
These conditions will be restated formally in the next section.

Operationally,
a cross-chain protocol is executed as follows.
Communication between parties and contracts occurs in synchronous rounds.  
Each round has four phases: 
\begin{enumerate}
\item 
  Parties optionally send messages to contracts;
\item
  The contacts receive messages (possibly sent in earlier rounds)
  and execute local computations;
\item 
  Parties read the contracts' new states;
\item 
  Parties execute local computations based on the states read and optionally prepare messages for the next round. 
\end{enumerate}
Message delivery may be delayed by congestion in the network, or in a chain's consensus protocol.
Different chains may accept messages at different rates.
(Most chains allow parties to pay higher fees to request speedier message delivery.)

Nevertheless,
we do assume that there is a known bound
on the number of rounds it takes for a message to be delivered on any chain.
It is usually convenient to express delay bounds in terms of time and clocks,
rather than rounds and round numbers.
Most blockchains produce new blocks at a more-or-less constant rate,
so the current block number is often used as a proxy for the current time.
We assume there is a known duration $\Delta > 0$
such that any message sent by a party to a contract
will be received, processed by the contract,
and observed by the other parties,
within duration $\Delta$.

The model encompasses synchronous execution with a known upper bound on message delivery delay
because this is the model currently presented by distributed ledgers such as
Bitcoin, Ethereum, Cardano, and others.
Timing upper bounds are important in practice because many decentralized finance protocols
 involve some form of \emph{escrow},
where a contract assumes temporary control of an asset that is intended to change hands.
(See, for example, the two-party swap example in \secref{example}.)
To be effective, each escrow needs an appropriate timeout:
if the exchange fails to complete in a reasonable amount of time,
the escrowed asset should be refunded to its original owner.
An asset should not be refunded prematurely (compromising liveness)
or held in escrow forever (compromising safety).

Because contracts cannot send messages to one another, they cannot directly observe one another's states.\footnote{In practice,
contracts \emph{on the same chain} can call one another, but we are interested in problems
where the contracts reside on distinct chains.}
Deviating parties may communicate with one another through channels
hidden from compliant parties.
Failure to send a protocol message is detected by a timeout. 
In  \secref{formal} we will define in more detail the model including timing assumptions.

\section{Formal Definitions}
\seclabel{formal}
This section presents
the \emph{smart contract model}.
It formally defines the notion of parties, contracts, tasks, and protocols.

\subsection{Cross-chain Systems}
\seclabel{crosschain_model}
 A \emph{cross-chain system} $CCS=(\cP,\cC)$ is composed of a finite set of \emph{parties} $\cP=\{P_1, \ldots, P_m\}$, $m \geq 2$ and a finite set of \emph{smart contracts} $\cC=\{C_1, \ldots, C_n\}$, $n \geq 2$. 
 
Parties and smart contracts are both modeled as \emph{interface automata}~\cite{AlfaroH01}.
 An interface automaton is a tuple $IA=(V, V^{init}, \cA^I, \cA^O, \cA^H, \cT)$  where:
 \begin{itemize}
     \item $V$ is a set of states;
     \item $V^{init} \subseteq V$ is a set of initial states;
     \item $\cA^I$, $\cA^O$, $\cA^H$ are mutually disjoint sets of input, output and internal actions;
     \item $\cT \subseteq V \times \cA \times V$ is the set of steps where $\cA=\cA^I \cup \cA^O \cup \cA^H$   
 \end{itemize}
An action $a \in \cA$ is \emph{enabled} at some state $v \in V$ if there is a step $(v,a,v^\prime) \in \cT$ for some $v^\prime \in V$.
An \emph{execution fragment}  of an interface automaton is a finite sequence of alternate states and enabled actions $v_0,a_0,v_1,a_1,v_2 \ldots v_{t-1},a_{t-1},v_t$ such that $(v_i,a_i,v_{i+1}) \in \cT$,$  \forall i \in \{0, \dots, t-1\}$.

 Let $IA_P=(V_P, V_P^{init}, \cA_P^I, \cA_P^O, \cA_P^H, \cT_P)$ be the interface automaton of party $P \in \cP$ and $IA_C=(V_C, V_C^{init}, \cA_C^I, \cA_C^O, \cA_C^H, \cT_C)$ be the interface automaton of a smart contract $C \in \cC$. 
 Let $Shared(\cA_P, \cA_C)=\cA_P \cap \cA_C$ be the common actions of $IA_P$ and $IA_C$. 
 $IA_P$ and $IA_C$ are \emph{composable} if the following four properties hold:
 \begin{align*}
   \cA_P^I \cap \cA_C^I &= \emptyset\\
   \cA_P^O \cap \cA_C^O &= \emptyset\\
   \cA_P^H \cap \cA_C &= \emptyset\\
   \cA_C^H \cap \cA_P = \emptyset
 \end{align*}
 The \emph{product}  of two composable interface automata $IA_P$ and $IA_C$ is the interface automaton
 \begin{align*}
   IA_{P \otimes C}=(&V_P \times V_C,\\
                  &V_P^{init} \times V_C^{init},\\
                  &\cA_P^I \cup \cA_C^I \setminus  Shared(\cA_P, \cA_C),\\
                  &\cA_P^O \cup \cA_C^O \setminus  Shared(\cA_P, \cA_C)\\
                  &\cA_P^H \cup \cA_C^H \cup  Shared(\cA_P, \cA_C),\\
                  &\cT_P \otimes  \cT_C)
    \end{align*}
 where $\cT_P \otimes  \cT_C$ is defined as:
 \begin{equation*}
\cT_P \otimes  \cT_C = TR_P \cup TR_C \cup TR_{PC}
 \end{equation*}
and
 \begin{align*}
 TR_P &=\set{((v,u),a,(v^\prime,u)) \mid (v,a,v^\prime) \in \cT_P \wedge a \notin Shared(P,C) \wedge u \in V_C}\\
 TR_C &=\set{((v,u),a,(v,u^\prime))  \mid  (u,a,u^\prime) \in \cT_C \wedge a \notin Shared(P,C) \wedge v \in V_P}\\
 TR_{PC}&=\set{((v,u),a,(v^\prime,u^\prime))  \mid  (v,a,v^\prime) \in \cT_P \wedge (u,a,u^\prime) \in \cT_C \wedge a \in Shared(P,C)}
 \end{align*}
The interface automaton of a cross-chain system $CCS=(\cP,\cC)$ is 
the composition of the interface automata of parties in $\cP$ and interface automata of smart contracts in $\cC$ as proposed by Alfaro and Henzinger~\cite{AlfaroH01}.
 
 In the sequel if $V$ is a vector, $V[i]$ is $V$'s $i^\text{th}$ element.
If $D$ is a domain, $2^D$ is the powerset of $D$.

\subsection{Cross-chain Tasks}
 
 Given a cross-chain system $\text{CCS}=(\cP,\cC)$ where $\cP$ is a set of $m$ parties  and $\cC$ a set of $n$ smart contracts
 a \emph{cross-chain task} is a tuple $(\cI_P,\cI_C,\cO_C, U)$, where:
 \begin{itemize}
 \item 
   $\cI_P$ is a set of $m$-element \emph{input party vectors},
   representing each party's input to the task,
 \item
   $\cI_C$ is a set of $n$-element \emph{input contract state vectors},
   representing each contract's state before executing the task,
 \item
   $\cO_C$ is a set of $n$-element \emph{output contract state vectors},
   representing each contract's state after executing the task, and
 \item
   $U:  \cI_P \times \cI_C \times \cO_C \to \Reals^m$ is a \emph{utility function}
   that characterizes how each party values each possible transition.
 \end{itemize}
 The utility for an individual party $P$ is written $U(I_P,I_C,O_C)[P]$,
 and the utility for a coalition $\cQ \subset \cP$ is defined to be
 \begin{equation*}
   U(I_P,I_C,O_C)[\cQ] := \sum_{Q \in \cQ} U(I_P,I_C,O_C)[Q].
 \end{equation*}
As mentioned,
the utility function captures the notion of ``better off'' and ``worse off'' for parties.
A \emph{transition} is a triple $(I_P,I_C,O_C) \in \cI_P \times \cI_C \times \cO_C$.
Party $P$ considers transition $(I_P,I_C,O_C)$ \emph{acceptable}
if $U(I_P,I_C,O_C)[P] \geq 0$,
and \emph{preferred} if $U(I_P,I_C,O_C)[P] > 0$.
A transition is \emph{acceptable} if it is acceptable to all parties,
and \emph{preferred} if it is preferred by all parties.

Note that a task's set of output contract states must encompass all reachable output contract states,
even those produced when parties deviate.
Because parties are autonomous (and possibly Byzantine),
a task definition is not expressed in terms of parties' states, 
only in terms of the input values they provide to the task.

For a task to be \emph{feasible} (capable of solution),
it must satisfy certain additional feasibility conditions.
Because deviating parties can always obstruct progress,
the null transition must always be acceptable.
For any input vectors
$I_P \in \cI_P$ and $I_C \in \cI_C$,
\begin{equation}
  \eqnlabel{feas:nothing}
  \cI_C \subset \cO_C \text{ and } (I_P,I_C,I_C) \text{ is an acceptable transition}.
\end{equation}
Each task must be solvable in principle:
for all input vectors $I_P \in \cI_P, I_C \in \cI_C$,
there must exist a preferred transition:
\begin{equation}
  \eqnlabel{feas:prefer}
  \exists O_C \in \cO_C \text{ such that } (I_P,I_C,O_C) \text{ is a preferred transition}.
\end{equation}
A task will have no solution if a deviating party $Q$ can trick a compliant
party $P$ into negative utility simply by lying about $Q$'s input.
For any two input party vectors $I_P, I_P'$, and for every party $P$,
\begin{equation}
  \eqnlabel{feas:inputs}
  \text{if } I_P[P] = I_P'[P] \text{ and } U(I_P,I_C,O_C)[P] \geq 0,
  \text{ then } U(I_P',I_C,O_C)[P] \geq 0.
\end{equation}
A hidden input might shift a compliant party's utility from from one non-negative quantity to another,
but never from non-negative to negative.

\subsection{Cross-chain Protocols}
To execute a task,
parties agree on a sequence of contract calls called a \emph{cross-chain protocol}.

As noted,
compliant parties follow the agreed-upon protocol,
while deviating parties do not.

When describing a protocol execution,
it is convenient to indicate which parties are compliant by a
\emph{compliance set} $\cQ \subset \cP$,
where $P \in \cQ$ means $P$ is compliant.

Formally, given a cross-chain system $CCS=(\cP,\cC)$ where $\cP$ is a set of $m$ parties and $\cC$ is a set of $n$ smart contracts, a cross-chain 
protocol is a tuple $(\cI_P,\cI_C,\cO_C,\Xi)$,
where
\begin{itemize}
\item 
  $\cI_P$ is a set of $m$-element \emph{input party vectors},
\item
  $\cI_C$ is a set of $n$-element \emph{input contract state vectors},
\item
  $\cO_C$ is a set of $n$-element \emph{output contract state vectors},
\item
  $\Xi: \cI_P \times \cI_C \times 2^\cP \to 2^{\cO_C}$,
  the \emph{execution function}, is a map
  that carries a input party vector, an input contract state vector,
  and a compliance set
  to a set of output contract state vectors representing possible outcomes.
\end{itemize}
The protocol itself is the interface automaton obtained by the composition of interface automata modeling parties in $\cP$ and interface automata modeling smart contracts in $\cC$.

The execution of a cross-chain protocol is an execution fragment (see Section \secref{crosschain_model}) starting in a state $(I_P,I_C) \in \cI_P \times \cI_C$ and terminating in a state in $\Xi(I_P,I_C,\cQ)$ where $\cQ \subseteq \cP$ indicates which parties were compliant during the protocol execution.   

A cross-chain protocol for a cross-chain task is \emph{correct} if each execution of the protocol satisfies the following properties:
\begin{itemize}
\item 
\emph{Coalition Nash Equilibrium:}
No coalition of parties can increase its collective utility
by deviating from the protocol while the others comply.
More precisely,
for all coalitions $\cQ \subset \cP$,
input vectors $I_P \in \cI_P, I_C \in \cI_c$,
conforming executions $O_C \in \Xi(I_P, I_C, \cP)$,
and $\cQ$-deviating executions $O_C' \in  \Xi(I_P, I_C, \cP \backslash \cQ)$,
conforming is the better strategy for $\cQ$
\begin{equation*}
U(I_P,I_C,O_C)[\cQ] \geq U(I_P,I_C,O_C')[\cQ].
\end{equation*}        
\item
\emph{Liveness:}
If all parties are compliant
then the protocol's transitions are preferred:
\begin{equation}
  \eqnlabel{protocol:prefer}
  (\forall O_C \in \Xi(I_P,I_C,\cP))\;
  (\forall P \in \cP)\;
  U(I_P,I_C,O_C)[P] > 0
\end{equation}

\item
\emph{Safety:}
  No compliant party ends up worse off:
\begin{equation*}
(\forall O_C \in \Xi(I_P,I_C,\cQ))\;
  (\forall Q \in \cQ)\;
  U(I_P,I_C,O_C)[Q] \geq 0.
  \end{equation*}
\end{itemize}

\subsection{Communication and Timing}
In a cross-chain system $CCS=(\cP,\cC)$,
parties in $\cP$ communicate with contracts in $\cC$ via messages.
Compliant parties do not communicate directly with other parties,
and contracts do not communicate directly with contracts.
Communication channels are \emph{authenticated}: a message's sender cannot be forged,
and \emph{reliable}: a channel does not create, lose, or duplicate messages.

As outlined earlier,
executions proceeds in synchronous rounds.
\begin{enumerate}
\item
  In the \emph{send phase},
  each party optionally sends a message to one or more contracts.
\item
  In the \emph{contract-local phase},
  each contract receives the messages in an arbitrary order
  and undergoes state changes in response.
\item
  In the \emph{read phase}
  each party reads the new contract state.
\item
  In the \emph{party-local phase},
  each party optionally prepares a new message.
\end{enumerate}

\section{Case Study: Two-Party Cross-Chain Swap}
\seclabel{example}
We illustrate an application of the model by analyzing a simple
\emph{two-party cross-chain swap} task,
along with a protocol~\cite{tiersnolan}.

\subsection{The Task}
Imagine that Alice is willing to exchange her asset $a$
for Bob's asset $b$, and vice-versa.
They employ two contracts:
$C_A$ controls ownership of $a$ on one ledger,
and $C_B$ controls ownership of $b$ on another ledger.
Alice and Bob have access to a \emph{cryptographic hash function}
$H(\cdot)$ where it is infeasible to reconstruct $s$ given $H(s)$.

Formally, the swap task is given by $(\cI_P,\cI_C,\cO_C, U)$, where
\begin{itemize}
\item
  $\cI_P$ consists of a single input party vector $I_P$
  with an entry for Alice, $I_P[A] = (B,a,b)$,
  and an entry for Bob, $I_P[B] = (A,b,a)$,
  indicating that each party is willing to swap its asset for the other's.

\item
  $\cI_C$ consists of a single input contract state vector $I_C$, with
  $C_A$'s entry $I_C[C_A] = [a \hookrightarrow A]$, (meaning $A$ owns $a$), and
  $C_B$'s entry $I_C[C_B] = [b \hookrightarrow B]$, (meaning $B$ owns $b$).

\item
  $\cO_C$ consists of four possible output contract state vectors:
  \begin{equation*}
    O_C = \begin{cases}
    [a \hookrightarrow A, b \hookrightarrow A]&\text{$A$ owns $a,b$}\\
    [a \hookrightarrow B, b \hookrightarrow A]&\text{$A$ owns $b$, $B$ owns $a$}\\
    [a \hookrightarrow A, b \hookrightarrow B]&\text{$A$ owns $a$, $B$ owns $b$}\\
    [a \hookrightarrow B, b \hookrightarrow B]&\text{$B$ owns $a,b$}.
  \end{cases}
  \end{equation*}
 
\item
Alice's utility function is:
\begin{equation}
  \eqnlabel{aliceUtility}
  U(I_P,I_C,O_C)[A] =
  \begin{cases}
    2 &\text{if } O_C = [a \hookrightarrow A, b \hookrightarrow A]\\
    1 &\text{if } O_C = [a \hookrightarrow B, b \hookrightarrow A]\\
    0 &\text{if } O_C = [a \hookrightarrow A, b \hookrightarrow B]\\
    -1 &\text{if } O_C = [a \hookrightarrow B, b \hookrightarrow B].
  \end{cases}
\end{equation}
Alice strongly prefers the outcome where she acquires both assets,
but also prefers the outcome where the swap takes place.
She accepts the outcome where assets do not change hands,
and she considers herself worse off
if Bob acquires her asset while retaining his own.
Bob's utility function is symmetric.
\end{itemize}

\subsection{The Protocol}

The protocol proceeds in rounds:
\begin{enumerate}
\item\label{aliceEscrows}
  Alice creates secret $s$ and hashkey $h=H(s)$.
  She instructs $C_A$ to assume temporary ownership of $a$,
  to transfer $a$ to Bob if he produces $s$ within two rounds,
  and to refund $a$ to Alice if Bob fails to meet that deadline.

\item\label{bobEscrows}
  Bob verifies that Alice has escrowed $a$ at $C_A$,
  If so, he learns $h$,
  and he instructs $C_B$ to assume temporary ownership of $b$,
  to transfer $b$ to Alice if she produces a matching $s$ within one round,
  and to refund $b$ to Bob if Alice fails to meet that deadline.
  If Bob sees that Alice has not escrowed $a$, he exits the protocol.

\item\label{aliceClaims}
  Alice verifies that Bob has escrowed $b$ at $C_B$.
  If so, she sends $s$ to $C_B$, making $s$ public,
  and causing $C_B$ to transfer $b$ to her.
  If Alice observs that Bob has not escrowed, she exits the protocol.
  If, b the end of this round, $C_B$ did not receive $s$ from Alice,
  then $C_B$ refunds $b$ to Bob.

\item\label{bobClaims}
  If Alice took possession of $b$, then Bob learns $s$ from $C_B$.
  He sends $s$ to $C_A$,
  causing $C_A$ to transfer $a$ to him, completing the swap.
  If, by the end of this round, $C_A$ did not receive $s$ from Bob,
  then $C_A$ refunds $a$ to Alice.
\end{enumerate}

Formally, the protocol is given by $(\cI_P,\cI_C,\cO_C,\Xi)$, where:
\begin{itemize}
\item
  $\cI_P, \cI_C$ and $\cO_C$ are defined as in the swap task above.
\item
  Since the task has only a single input party vector
  and a single input contract state vector,
  we can restrict the execution function
  $\Xi: \cI_P \times \cI_C \times 2^\cP \to \cO_C$
  to its compliance set argument alone
  $\Xi: 2^\cP \to \cO_C$.
  \begin{itemize}
  \item
    If both parties are compliant, the swap takes place:
    \begin{equation}
      \eqnlabel{bothComply}
      \Xi(\set{A,B}) = [a \hookrightarrow B, b \hookrightarrow A].
    \end{equation}
  \item
    Alice can deviate in several ways:
    \begin{equation*}
      \eqnlabel{bobComplies}
      \Xi(\set{B}) = \begin{cases}
        [a \hookrightarrow A, b \hookrightarrow B]
          &\text{Alice fails to escrow $a$ at Step \ref{aliceEscrows}}\\
        [a \hookrightarrow A, b \hookrightarrow B]
          &\text{Alice fails to claim $b$ at Step \ref{aliceClaims}}\\
        [a \hookrightarrow B, b \hookrightarrow B]
          &\text{Alice reveals $s$ to Bob after she escrows $a$, without claiming $b$}
      \end{cases}
    \end{equation*}

  \item
    Bob can deviate in several ways:
    \begin{equation*}
      \eqnlabel{aliceComplies}
      \Xi(\set{A}) = \begin{cases}
        [a \hookrightarrow A, b \hookrightarrow B]
          &\text{Bob fails to escrow $b$ at Step \ref{bobEscrows}}\\
        [a \hookrightarrow A, b \hookrightarrow A]
          &\text{Bob fails to claim $a$ at Step \ref{bobClaims}}
      \end{cases}
    \end{equation*}
  \end{itemize}

\item
  If both parties deviate, $\Xi(\emptyset)$ contains all four outcomes.
\end{itemize}

\subsection{The Proof}

\begin{lemma}[Liveness]
  \lemmalabel{liveness}
  If both parties are compliant,
  then for $O_C \in \Xi(\cP)$, 
  and $P \in \cP$,
  $U(I_P,I_C,O_C)[P] > 0$.
\end{lemma}
\begin{proof}
  Recall that $\Xi(\cP) = [a \hookrightarrow B, b \hookrightarrow A])$,
  which has utility 1 for both Alice and Bob (\eqnref{aliceUtility}).
\end{proof}

\begin{lemma}[Safety]
  \lemmalabel{safety}
  No compliant party ends up with negative utility.
\end{lemma}
\begin{proof} 
By case analysis:
\begin{itemize}
\item
  If both parties are compliant,
  both have positive utility by \lemmaref{liveness}.

\item
  If Alice alone is compliant,
  \begin{equation*}
    \Xi(\set{A}) = \set{
      [a \hookrightarrow A, b \hookrightarrow B],
      [a \hookrightarrow A, b \hookrightarrow A]
    }
  \end{equation*}
  Each of these outcomes has non-negative utility for Alice (\eqnref{aliceUtility}).

\item
  If Bob alone is compliant,
  \begin{equation*}
    \Xi(\set{B}) = \set{
      [a \hookrightarrow A, b \hookrightarrow B],
      [a \hookrightarrow B, b \hookrightarrow B]}
  \end{equation*}
  Each of these outcomes has non-negative utility for Bob.
\end{itemize}    
\end{proof}

\begin{lemma}[Coalition Nash Equilibrium]
  \lemmalabel{equilibrium}
  No coalition of parties can increase its collective utility
  by deviating from the protocol while the others comply.
\end{lemma}
\begin{proof}
Since there are only two parties, of which one is compliant,
coalitions are singletons.
  \begin{itemize}
  \item
    Suppose the coalition is Alice.
    If both comply (\eqnref{bothComply}),
    Alice's utility is 1 (\eqnref{aliceUtility}).
    If Alice deviates while Bob complies (\eqnref{bobComplies}),
    Alice's utility is either 0 or -1 (\eqnref{aliceUtility}).

  \item
    Suppose the coalition is Bob.
    If both comply (\eqnref{bothComply}), Bob's utility is 1.
    If Bob deviates while Alice complies (\eqnref{aliceComplies}),
    Bob's utility is either 0 or -1.
  \end{itemize}
\end{proof}

\section{Conclusions}
\seclabel{conclusion}
This paper has proposed the \emph{smart contract} model
for
cross-chain protocols in which  parties, some honest and some not, interact through trusted smart contracts residing on multiple, independent ledgers.
As a case study, it  presented a simple two-way  cross-chain swap protocol, 
solving a fundamental task in this setting.
Many other tasks have been considered in real systems that would be interesting to formalize
in the smart contract model, such  as auctions, loans and options.
  
The smart contract model differs from classical
models of distribute computing,
refactoring the roles of active participants (processes, parties)
and of passive communication (objects, smart contracts).
%SR: not accurate, many papers have considered Byzantine failures
%In the classical distributed models,
%active participants are usually prone to crashes but not to Byzantine failures.
In classical models,
Byzantine processes interact with Byzantine processes,
while in the smart contract model,
Byzantine parties interact with honest contracts.
The  model requires a nuanced, game theoretic notion of correctness distinct from that of classical models.

The model assumes smart contracts cannot communicate with one another,
because of the lack of a practical way
to ensure that replaying  such cross-contract communication would always
and everywhere produce exactly the same results.
Nevertheless,
there are emerging real-world mechanisms that do support limited
forms of cross-chain communication.
An \emph{Oracle}~\cite{Chainlink} is a mechanism that allows a contract to read data
from an external source (for example, the current dollar/euro exchange rate).
Token bridges~\cite{axelar,Wormhole,layerzero} provide a way to effectively
transfer assets from one ledger to another by ``freezing'' an asset
at the source ledger and creating a matching ``wrapped'' asset at the target ledger.
So-called \emph{layer two} solutions~\cite{rollups,KalodnerGCWF2018,plasma}
allow computations to be moved from one ledger to another.
We leave it to future work to expand the model presented here to encompass
these increasingly important technologies.

The model's computational power can vary depending on access to
cryptographic primitives.
For example,
cryptographic hashes enable atomic cross-chain swaps~\cite{bitcoinwiki,bip199,decred,Herlihy2018,tiersnolan,barterdex,ZakharyAE2019,Catalyst} and payment networks~\cite{DeckerW2015,bolt,HeilmanLG2019,raiden,PoonD2016}.
Public key infrastructure supports atomic broadcast~\cite{Herlihy2018,HerlihyLS2021}.
To fully capture both the distributed and cryptographic aspects of these applications
it might be necessary to extend the model in a way similar to Canetti et al.~\cite{CanettiCKLP07}.
Moreover,
it might be intriguing to extend the model to encompass probabilistic behavior,
perhaps using extensions to I/O automata similar to those of Civit and Potop~\cite{CivitP22b}.

\bibliographystyle{splncs04}
\bibliography{references,zotero,url}

%\newpage
\appendix

\end{document}